\documentclass[manyauthors]{fundam}
\usepackage[utf8]{inputenc}
\usepackage{amssymb}
\usepackage{graphicx}

\newcommand{\pre}[1]{^\bullet #1}
\newcommand{\post}[1]{#1 ^\bullet}
\newcommand{\cnfo}{\mathrel{\natural}}
\newcommand{\conc}{\textbf{co}}

\DeclareMathOperator{\ltlx}{\mathrm{LTL-X}}
\DeclareMathOperator{\unf}{\textsc{unf}}
\DeclareMathOperator{\atlx}{\mathrm{1ATL-X}}
\DeclareMathOperator{\atls}{\mathrm{ATL*}}
\DeclareMathOperator{\MG}{\textsc{mg}}  % Marking graph
\DeclareMathOperator{\ts}{\textsc{ts}}  % Labelled transition system
\title{Looking for winning strategies in two-player games on Petri nets with partial observability}
\author{Federica Adobbati\corresponding\\DISCo, Univ. di Milano-Bicocca\\viale Sarca 336-U14, Milano, Italy\\
f.adobbati@campus.unimib.it
\and Luca Bernardinello\\DISCo, Univ. di Milano-Bicocca\\viale Sarca 336-U14, Milano, Italy\\
luca.bernardinello@unimib.it\\
\and Lucia Pomello\\DISCo, Univ. di Milano-Bicocca\\viale Sarca 336-U14, Milano, Italy\\
lucia.pomello@unimib.it} 
\date{}
\begin{document}
\runninghead{F. Adobbati, L. Bernardinello, L. Pomello}{A game on Petri nets with partial observability}
\maketitle
\begin{abstract}
  We define a game on 1-safe Petri nets, where a \emph{user} plays 
  against an \emph{environment} in order to reach a goal on the 
  system. The goal is expressed through an $\ltlx$ formula, and 
  represents a behaviour of the system that the user needs to 
  guarantee. The user can try to reach his goal by controlling a 
  subset of transitions, and by observing a subset of local states. 
  Although we do not put any requirement on the local states 
  observable by the user, we assume that he cannot be sure to 
  observe them in the exact moment in which they become marked. 
  For this reason, we define a notion of \emph{stability} of the 
  observation. 
  We propose a method to determine whether the user has a 
  \emph{winning strategy}, i.e. if he can win every play by 
  taking some decisions based on the information available for him. 
\end{abstract}
\begin{keywords}
Petri nets, distributed games, temporal logics, partial observability
\end{keywords}
\section{Introduction}
Games can be used as a formal tool to model the interaction 
between agents. 
In this paper we define a game between a \emph{user} and 
an \emph{environment} on a distributed system modelled 
as a 1-safe Petri net. 
The task of the user is to control the system, so that only 
executions satisfying a desired property are allowed. 
In general, the environment does not cooperate with the 
system, and may be hostile or indifferent. 

In order to reach his goal, the user can control the 
occurrence of a subset of transitions, i.e. he can decide 
whether to fire them or not, when they are enabled. 
The rest of transitions belongs to the environment, and it 
is uncontrollable by the user. 
We assume a fairness constraint for the environment: 
no uncontrollable transition can be permanently enabled in 
the system and never fire. 
A \emph{play} is represented through a run on the unfolding 
of the net. 
A run represents an execution of the system, and records 
all the occurrences of both controllable and uncontrollable 
transitions. 
The user wins the play if the run satisfies the property 
that he needs to guarantee. 

In order to reach his goal, the user must decide the  transitions to 
fire based on his information on the current state of the 
system.  
We assume that he has \emph{no memory}, hence his decisions 
cannot be based on the states previously crossed during the 
play, and that the user has \emph{partial observability} 
on the system, i.e. he can never know the value of some of 
the local states. 
Even for observable states, we assume that the 
information about their values may take some time to 
arrive, and therefore the user can use them to take  
decisions only if they are \emph{stable}. 
This is motivated by the presence of different components 
in a distributed and concurrent system: in a realistic 
model, each communication from a component to another 
takes some time, and while the communication happens, the 
information may become outdated, due to the occurrence of 
a new transition.
Finally, we assume that the structure of the system is 
known, and the user can exploit all the information that he 
can derive from it. 
In particular, we will show in Sec.~\ref{sec:pn} and \ref{sec:game_net} 
how to add some places to the net that do not alter its 
behaviour, but allow the user to take his decisions 
combining the information derived from his observations 
and from the structure of the system. 

Our goal is to determine whether the user has a \emph{winning strategy}, i.e. if he can select transitions based on his 
observations that make him always win. 
We consider the case in which the goal of the user can 
be expressed through a fragment of LTL, and 
we propose a method based on the marking graph to find a 
winning strategy, if one exists.

Possible applications of such a game could be 
in the frame of verification, adaptation and control of distributed systems; 
so that, in the case of a winning strategy for the user with respect to a behavioural property, 
the system model could be adapted imposing that specific property, 
for example by adding a controller component which implements the user behaviour 
by synthesising the winning strategy; 
a reference for this sort of applications is for example \cite{RW87}.

The rest of the paper is organized as follows. Sec.~\ref{sec:pn} 
presents the basic definitions related to Petri nets, 
marking graphs, and unfoldings; in Sec.~\ref{sec:game_net}, 
we introduce the game formally, and discuss some examples 
to motivate our modelling choices. 
Sec.~\ref{sec:pn2cgs} and Sec.~\ref{sec:algo} provide the 
procedure to check whether the user has a winning strategy, 
and show a relation between our game and those 
defined on concurrent structures. 
In Sec.~\ref{sec:art} we discuss some related works.
Finally, Sec.~\ref{sec:concl} ends the paper and suggests 
future developments of this works.
\section{Petri nets} 
\label{sec:pn}
This section recalls some basic definitions and features of Petri nets, which will be used throughout the paper. As main reference see \cite{m89}.
Petri nets were introduced by C.A. Petri as a formal tool to represent
flows of information in distributed systems. 
In the last decades, several classes of nets have been defined and studied. 
In this work we use the class of \emph{1-safe} nets.

\begin{definition} %net%
A \emph{net} is a triple $N=(P, T, F)$, where $P$ and $T$  are disjoint sets.
The elements of $P$ are called \emph{places} and are represented by circles, 
the elements of $T$ are called \emph{transitions} and are represented by squares.
$F$ is called \emph{flow relation}, with
$F \subseteq (P \times T )\cup (T \times P)$, and is represented by arcs.
\end{definition}
For each element of the net $x \in P \cup T$, the \emph{pre-set} of $x$
is the set $\pre{x} = \{ y \in P \cup T \mid (y,x) \in F \}$, 
the \emph{post-set} of $x$ is the set
$\post{x} = \{ y \in P \cup T \mid (x,y) \in F\}$.
The previous notation can be extended to subsets of elements $A \subseteq P \cup T$: 
$\pre{A} = \bigcup_{x \in A} \ \pre{x}$ and $\post{A} = \bigcup_{x \in A} \post{x}$.
We assume that
each transition has non-empty pre-set and post-set.

Two transitions, $t_1$ and $t_2$, are \emph{independent}  if
$(\pre{t_1} \cup \post{t_1})$ and $(\pre{t_2} \cup \post{t_2})$
are disjoint. They are in \emph{conflict}, denoted with $t_1\# t_2$, if $\pre{t_1} \cap \pre{t_2} \neq \emptyset$.

A net $N'= ( P', T', F')$ is a \emph{subnet} of $N = (P, T, F)$ if 
$P' \subseteq P$, $T', \subseteq T$, and 
$F'$ is $F$ restricted to the elements in $N'$.

\begin{definition} % elementary net system
	A \emph{net system} is a quadruple
	$\Sigma = (P, T, F, m_{in})$
	consisting of a finite net $N = (P,T,F)$ and an
	\emph{initial marking}  $m_{in}: P \rightarrow \mathbb{N}$.
\end{definition}
A transition $t$ is \emph{enabled} at a marking $m$, denoted $m[t\rangle$,
if, for each $p\in {\pre{t}}$, $m(p) > 0$.
A transition $t$, enabled at $m$, can  \emph{occur} (or \emph{fire}) 
producing a new marking $m'$ (denoted $m[t\rangle m'$), where
\[
    m'(p) =
      \begin{cases}
        m(p) + 1 & \text{if } p \in t^\bullet \setminus {}^\bullet t\\
        m(p) - 1 & \text{if } p \in {}^\bullet t \setminus t^\bullet\\
        m(p)     & \text{otherwise}
      \end{cases}
\]
A marking $m'$ is reachable from another marking $m$, 
if there is a sequence of transitions $ t_1 t_2 \dots t_n$ such that 
$m[t_1 \rangle m_1 [ t_2 \rangle \dots m_{n-1}[ t_n \rangle m'$, this is also denoted with 
$m[t_1 t_2 \dots t_n \rangle m'$;
$[m \rangle$ denotes the set of markings reachable from $m$.
A marking $m$ is 
\emph{reachable} if it is reachable from the initial marking $m_{in}$, i.e.: if $m \in [m_{in} \rangle$; $[m_{in} \rangle$ will be also denoted $M$ in the next sections.

A net system is \emph{1-safe} if $m(p) \leq 1$, 
for each place $p$ and for each reachable marking $m$. 
Markings in 1-safe net systems
can, and will be, considered as subsets of places.

In a net system, two transitions, $t_1$ and $t_2$, are
\emph{concurrent} at a marking $m$
if they are independent and both enabled at $m$.

\begin{example}
	Fig.~\ref{fig:rapina} shows a net system with initial marking $m_{in} = \{u, r\}$. $sc_1$ and $sc_2$ are transitions which are in conflict; whereas $sc_1$ and $t$ are independent and both enabled at $m_{in}$, hence they are concurrent at $m_{in}$.
	\end{example}
\vspace{\baselineskip}

\noindent
The sequential behaviour of a net system 
can be described by an initialized labelled transition system, 
where the states correspond to reachable markings, 
with the initial one corresponding to the initial marking, 
and each arc is labelled by the transition leading from the source marking to the target one.

\begin{definition} % marking graph
	Let $\Sigma = (P, T, F, m_{in})$ be a  net system. 
	Its \emph{sequential marking graph} is a quadruple $\MG(\Sigma) = (M, T, A, m_{in})$
	where $M$ is the set of reachable markings in $\Sigma$, and where
	$A= \{ (m, t, m')\ | \ m, m' \in M, \  t \in T \ and \  m[t \rangle m' \}$.
\end{definition}
\vspace{\baselineskip}
% -------------------------------------------
\subsection{Implicit places}\label{s:implicit}
From now on, we will work on 1-safe net systems.
As remarked above, each node of the marking graph of a 1-safe net is a set of places.
Hence, we can associate to each place its \emph{extension}, namely
the set of reachable markings to which it belongs:
\[
    \forall p\in P\quad r(p) = \{ m\in \MG(\Sigma) \mid m(p) = 1 \}
\]
The extension of a place satisfies a few properties, which are
immediate consequences of the firing rule of nets, together with
1-safeness. The main property we need here can be called the
\emph{uniform crossing} property: for any given transition $t$, if
one occurrence of $t$ in the marking graph enters $r(p)$, then
all occurrences of $t$ do the same, and analogously when
an occurrence leaves $r(p)$. This property can be defined more
abstractly in the context of general labelled transition systems.

A labelled transition system is a structure $\ts = (Q,L,\Delta)$, where
$Q$ is the set of states, $L$ a set of \emph{labels}, and
$\Delta\subseteq Q\times L\times Q$ is the set of edges (edges are usually
called \emph{transitions}, but here we prefer to reserve this term
to the transitions of net systems).
A subset $r$ of states satisfies the uniform crossing property for
$\ts$ if, for all $t \in L$,
\begin{align*}
    (q_1,t,q_2), (q_3,t,q_4)\in \Delta, q_1 \in r, q_2 \not\in r\quad
    \Longrightarrow\quad q_3\in r, q_4 \not\in r\\
    (q_1,t,q_2), (q_3,t,q_4)\in \Delta, q_1 \not\in r, q_2 \in r\quad
    \Longrightarrow\quad q_3\in r, q_4 \not\in r
\end{align*}
A subset of states which
satisfies the uniform crossing properties is called a \emph{region}
(refer to~\cite{BBD15} for an extensive report on region theory, with
proofs of the properties recalled in the rest of this section). A region $r$
has a set of entering labels, denoted by $\pre{r}$, and a set of exiting labels,
denoted by $\post{r}$: if $r$ is a region and
$t\in L$ a label, then $t \in {\pre{r}}$ if there is an edge $(q_1,t,q_2)$
entering $r$, namely $q_1 \not\in r$ and $q_2 \in r$; symmetrically,
$t \in \post{r}$ if there is an edge leaving $r$.

If $r$ is a region
of the marking graph of a net system $\Sigma$, then
we can add a corresponding place to the net without changing the
behaviour: the marking graph of the extended net system is isomorphic
to the marking graph of $\Sigma$. 
The place $r\in P$ is a precondition of the all the transitions $t\in T$ such that $t\in \post{r}$, and 
a postcondition of the transitions $t\in T$ such that $t\in \pre{r}$.

The set of regions of a labelled transition system satisfies some useful
combinatorial properties.
We say that two regions, $r_1$ and $r_2$, are \emph{compatible}
if there exist three, pairwise disjoint, regions, $u_1, u_2, u_3$
such that $r_1 = u_1 \cup u_2$ and $r_2 = u_2 \cup u_3$.
If $r$ is a region, then its set-theoretic complement
$Q\setminus r$ is a region. 
In general, regions are not closed by union and
intersection, but the union of compatible regions is always a region.
In particular, the union of disjoint regions is a region. 
Since each region can be associated to a (potential) place on the net, these same  
properties can be equivalently stated between places of $\Sigma$. 
In particular, for each place $p\in P$ with extension $r(p)$, 
its complement $p^c$ is the place with extension $Q \setminus r(p)$; 
two places $p_1, p_2 \in P$ are \emph{compatible}, denoted with $p_1 \$ p_2$, 
iff their extensions $r(p_1)$ and $r(p_2)$ are compatible; 
the union of two compatible places $p_1$ and $p_2$ is the place with 
extension $r(p_1) \cup r(p_2)$ and is denoted with $p_1 \lor p_2$. 

We will call \emph{implicit place} of a net system $\Sigma$ a region
of $\MG(\Sigma)$ which is not the extension of a place of $\Sigma$.
Let $H$ be a set of implicit places of $\Sigma$; then $\Sigma_H$
denotes the net system obtained from $\Sigma$ by adding the implicit
places in $H$. For each place $h$ in $H$ we add an arc to the flow
relation from $t$ to $h$ for each $t$ in $\pre{h}$ and an arc from
$h$ to $t'$ for each $t'$ in $\post{h}$. An implicit place belongs to
the initial marking of $\Sigma_H$ if $m_0$ belongs to the corresponding
region.
\subsection{Branching processes and unfoldings}
The non sequential behaviour of 1-safe net systems can be recorded 
by occurrence nets, which are used to represent by a single object 
the set of potential histories of a net system.
In the following, we denote with $F^+$ the transitive closure of $F$, and with $F^*$ the reflexive closure of $F^+$. 

Given two elements $x, y \in P \cup T$ we write $x \cnfo y$, if there exist
$t_1, t_2 \in T : t_1 \neq t_2, t_1 F^* x, t_2 F^* y \land
\exists p \in {}^{\bullet}t_1 \cap {}^{\bullet}t_2$.

\vspace{\baselineskip}
\begin{definition} %occurrence  net 
A net $N =(B, E, F)$ is an \emph{occurrence net} if
\begin{itemize}
	\item for all $b \in B$, $|\pre{b}| \leq 1$
	\item $F^*$ is a partial order on $B \cup E$
	\item for all $x \in B \cup E$, the set
	$\{ y \in B \cup E \mid y F^* x \}$ is finite
	\item for all $x \in B \cup E$, $x \cnfo x$ does not hold
\end{itemize}
\end{definition}

We will say that two elements $x$ and $y$, $x \neq y$,
of $N$ are \emph{concurrent},
and write $x\, \conc\, y$, if they are not ordered by $F^*$,
and $x \cnfo y$ does not hold.

By $\min(N)$ we will denote the set of minimal elements with
respect to the partial order induced by $F^*$.

A \emph{B-cut} of $N$ is a maximal set of pairwise concurrent
elements of $B$. B-cuts represent potential global states
through which a process can go in a history of the system. 
By analogy with net systems, we will sometimes
say that an event $e$ of an occurrence net is \emph{enabled}
at a B-cut $\gamma$, denoted $\gamma[e\rangle$, if $\pre{e} \subseteq \gamma$.
We will denote by $\gamma + e$ the B-cut
$(\gamma\backslash ^\bullet e) \cup e^\bullet$.
A B-cut is a \emph{deadlock cut} if no event is enabled at it.

Let $\Gamma$ be the set of B-cuts of  $N$.
A partial order on $\Gamma$ can be defined as follows: let
$\gamma _1, \gamma _2$ be two B-cuts. We say
$\gamma _1 \leq \gamma _2$ iff
\begin{enumerate}
	\item $\forall y \in \gamma _2 \ \exists x \in \gamma _1: xF^*y$
	\item $\forall x \in \gamma _1 \ \exists y \in \gamma _2: xF^*y$
\end{enumerate}
In words, $\gamma_1 \leq \gamma_2$ if any condition
in the second B-cut is or follows a condition of the first B-cut and
any condition in the first B-cut is or comes before a condition of the
second B-cut. $\gamma_1 < \gamma_2$ if the conditions 
above hold and  $\exists x\in \gamma _1$,$\exists y \in \gamma _2 : xF^+y$.

A sequence of B-cuts, $\gamma _0 \gamma _1 \dots \gamma _i \dots$
is \emph{increasing} if $\gamma_i < \gamma_{i+1}$ for all $i \geq 0$.
A cut $\gamma$ is \emph{compatible} with an increasing 
sequence of B-cuts $\delta$ iff there are two cuts 
$\gamma_i, \gamma_{i+1}\in \delta$ such that $\gamma_i\leq \gamma \leq \gamma_{i+1}$. 

Given an increasing sequence of B-cuts $\delta$, we define a \emph{refinement} 
$\delta'$ of $\delta$ as an increasing sequence of B-cuts such that for each 
$\gamma\in\delta$, $\gamma$ is also in $\delta'$. 
A \emph{maximal} refinement $\delta'$ is an increasing sequence of 
B-cuts such that there is no $\gamma\not\in\delta'$ compatible with 
$\delta'$.
\begin{example}
Fig.~\ref{fig:play} shows an occurrence net and 
an increasing sequence of B-cuts, represented 
with bold lines. The B-cut $\{b_9, b_{10}\}$ % $\{s_{2,1}, m_{2,0}\}$ 
is compatible with the sequence in the figure, 
therefore the sequence including the three 
represented B-cuts and $\{b_9, b_{10}\}$ % $\{s_{2,1}, m_{2,0}\}$ 
is a refinement of the sequence in figure.
\end{example}
We say that an event $e \in E$ precedes a B-cut $\gamma$,
and write $e < \gamma$,
iff there is $y \in \gamma$ such that $eF^+y$.
In this case, each element of $\gamma$ either follows $e$ or is
concurrent with $e$ in the partial order induced by the occurrence net. 

The next definitions are adapted from~\cite{Ebp91}. 
\begin{definition} %branching process
A \emph{branching process} of a 1-safe net system $\Sigma = (P, T, F, m_{in})$
is an occurrence net $N = (B, E, F)$, together with a
labelling function $\mu: B \cup E \rightarrow P \cup T$,
such that
\begin{itemize}
	\item $\mu(B) \subseteq P$ and $\mu(E) \subseteq T$
	\item for all $e \in E$, the restriction of $\mu$ to $\pre{e}$
	is a bijection between $\pre{e}$ and $\pre{\mu(e)}$;
	the same holds for $\post{e}$
	\item the restriction of $\mu$ to $\min(N)$ is a bijection
	between $\min(N)$ and $m_{in}$
	\item for all $e_1, e_2 \in E$, if $\pre{e_1} ={}\pre{e_2}$ and
	$\mu(e_1) = \mu(e_2)$, then $e_1 = e_2$
\end{itemize}
A \emph{run} of $\Sigma$ is a branching process $(N, \mu)$
such that there is no pair of elements $x, y$ in $N$ such that 
$x \cnfo y$.
\end{definition}
For $\gamma$ a B-cut of $N$, the set $\{\mu(b) \mid b \in \gamma\}$
is a reachable marking of $\Sigma$ (\cite{Ebp91}), and we refer to it as the
marking corresponding to $\gamma$.

Let $(N_1, \mu_1)$ and $(N_2, \mu_2)$ be two branching processes of $\Sigma$, 
where $N_i = (B_i, E_i, F_i)$, $i = 1,2$.
We say that $(N_1, \mu_1)$ is a \emph{prefix}
of $(N_2, \mu_2)$ if $N_1$ is a subnet of $N_2$, and 
$\mu_2|_{B_1\cup E_1} =  \mu_1$. 
For any net system $\Sigma$, there exists a unique,
up to isomorphism, maximal branching process of $\Sigma$.
We will call it the \emph{unfolding} of $\Sigma$,
and denote it by $\unf(\Sigma)$ .

A \emph{run} of $\Sigma$ describes a particular
history of $\Sigma$, in which conflicts have been solved.
Any run of $\Sigma$ is a prefix of the unfolding $\unf(\Sigma)$;
we also say that it is a run on $\unf(\Sigma)$.
\section{A two-player game on a 1-safe Petri net}
\label{sec:game_net}
In this section, we consider 1-safe net systems (in the following called only net systems)
in which a subset of transitions 
is \emph{controllable} by a \emph{user}. 
This means that whenever a controllable transition is enabled, the user 
can decide whether to fire it or not. All the other transitions belong
to the \emph{environment}. 
In figures, controllable transitions will be 
represented in dark grey, while environment transitions will be white. 

The goal of the user is to force the behaviour of the system so that every 
execution satisfies a certain behaviour (e.g. reaching and/or avoiding some 
subsets of places). 
In order to reach his goal, the user can observe part of the current 
global state and decide whether to fire some controllable enabled transitions. 
If a controllable transition is in conflict with an uncontrollable transition, 
and they are both enabled, we assume that the user cannot guarantee to fire 
his transition before the environment. 
We model the interaction between the user and the 
environment as a game on the net, where the user and the environment 
are two players challenging each other. 
Informally, a play is an execution on the net, and the user 
wins a play if it satisfies his goal. 
If the user is always able to force the behaviour of the system 
according to his goal, and therefore to win all the plays 
that he enforces from the initial state, 
we say that the user has a \emph{winning strategy} for that goal. 
We assume that the user has \emph{no memory}, 
i.e. his decisions can only be based on the current state of 
the system, and not on its history.

Formally, let $\Sigma = (P, T, F, m_{in})$ be a net system, 
$T_u$ the set of transitions controllable by the user, and  
$T_{env} = T\setminus T_u$ the set of transitions belonging to the 
environment. 
We assume that the user observes a subset of places $P_O \subseteq P$, 
and that he observes which of his transitions are enabled, i.e. 
$^\bullet T_u \subseteq P_O$. 
In figures, observable places are represented with thicker borders. 

In addition to the observability, in order to define what information 
the user can exploit to reach his goal, we need to introduce the 
concept of \emph{stable part} of a marking. 
Let $m\in [m_{in}\rangle$ be a reachable marking. Its stable part is 
defined as the set $m^s = \{p\in m \mid \nexists t_1t_2...t_n\in T^*_{env} : p\in {^\bullet t_n} \wedge (m[t_1...t_{n-1}\rangle m_n \wedge {^\bullet t_n}\subseteq m_n) \}$; 
in other words $m^s$ is the set of places in $m$ that cannot be 
consumed by any sequence of uncontrollable transitions enabled in $m$. 
In what follows, we assume that the user can base his decisions 
only upon the observable and stable parts of markings. This is 
motivated by the fact that in realistic scenarios, even if 
the value of a certain local state is not hidden to the user, 
its information may arrive to the user with a certain delay. 
By definition, if a place is not in the stable part of a 
marking, its value may change due to a transition that is out 
of the control of the user, therefore the information about it 
may arrive outdated to the user, that cannot count on it. 
To clarify this concept, consider the following example.
\begin{figure}[!ht]
    \centering
    \includegraphics[width=0.7\textwidth]{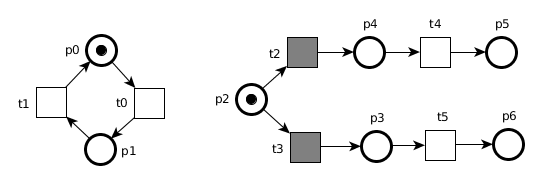}
    \caption{A net system with controllable and uncontrollable transitions 
    (coloured resp. in grey and white)}
    \label{fig:stable}
\end{figure}
\begin{example}
\label{ex:stable}
In the net in Fig.~\ref{fig:stable}, 
assume that $P_O = P$ and that the goal of the user is to reach 
one of the two markings $\{p0, p4\}$ or $\{p1, p3\}$. 
If the user could base his decision on the global states, then 
he could reach his goal by firing $t2$ when the marking is 
$\{p0, p2\}$ or $t3$ when the marking is $\{p1, p2\}$. 
However, even if the user observes an occurrence of $p0$, 
it is unrealistic to assume that the communication arrives without time delay 
and that he succeeds in firing $t2$ before the environment can 
fire $t0$.
In our setting, the user can base his decisions only on the stable 
part of the markings. For both $\{p0, p2\}$ and $\{p1, p2\}$ 
the stable part is $\{p2\}$, therefore the user cannot 
differentiate between the two markings, and cannot reach his goal. 
\begin{figure}
    \centering
    \includegraphics[width = 0.5\textwidth]{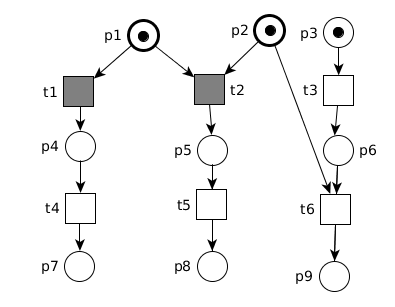}
    \caption{A net system with no stable marking in which $p2$ is enabled}
    \label{fig:stpart}
\end{figure}

In some cases, the definition of stable part of a marking has consequences 
also on the controllable transitions that the user can fire. 
Consider the net system in Fig.~\ref{fig:stpart}, where $p1$ and $p2$ 
are the observable places, and $t1$ and $t2$ are the 
controllable transitions. 
Assume that the goal of the user is to reach $p8$. Intuitively, from the 
initial marking the user cannot prevent the sequence $t3t6$ to fire, and 
therefore he cannot count on the occurrence of $t2$ to win, even if $t2$ 
is controllable and enabled in the initial marking. 
This can be modelled thanks to the definition of stable part of a marking: 
there is no marking where $p2$ is stable, because from every marking containing 
$p2$ there is a sequence %set 
of uncontrollable transitions that could fire and make $p2$ 
false (for example, as we already observed, the sequence $t3t6$ in the initial marking).
\end{example}
In some cases, the information coming from the structure of the net 
and from some observations, even if delayed, is sufficient for 
the user to reach his goal, but is not included in the stable 
parts of the reachable markings. 
For example, the user may know that an observable place cannot be 
marked anymore in the future or that there must be a token in a certain 
set of observable places, even without knowing precisely in which one (see Ex.~\ref{ex:sature}). 
For this reason, in order to use this information in the strategy, 
we consider some implicit places, corresponding to unions and complements of observable 
and compatible places. 
In particular, the set $P'$ of places that the strategy will consider is defined as follows. 
All the places in $P$ are also in $P'$, i.e. $P\subseteq P'$. 
We denote the set of observable places in $P'$ as $P_O'$. 
$P_O'$ is constructed recursively as follows. 
$P_O \subseteq P'_O$; if $p\in P'_O$, then its complement $p^c$ is in $P'_O$; if 
$p_1, p_2$ are in $P'_O$ and $p_1 \$ p_2$ then $p_1 \lor p_2\in P_O'$. 
By construction all the places in $P' \setminus P$ are observable. 
In the figures, the implicit places will be coloured in light grey.
We denote the net with these additional places as $\Sigma' = (P', T, F', m'_{in})$; in what follow 
$\Sigma'$ will be also referred as \emph{extension} of $\Sigma$ or \emph{extended net}.

Recall, from Sec.~\ref{s:implicit}, that adding implicit places
to a net does not change its behaviour. The extended net just
makes explicit an information which was actually available in
the original net. On the other hand, in defining a strategy
we cannot use information corresponding to the union, or
logical disjunction, of non compatible places, because this
information cannot be coded in the form of a place,
explicit or implicit. 

Let $M^S$ be the set of stable parts of the markings in $[m'_{in}\rangle$; for each $m^s\in M^S$, $o(m^s) = m^s \cap P'_O$, 
i.e. $o(m^s)$ is the set of observable places in the stable part 
of a certain marking. 
We denote $O(M^S) = \{o(m^s) \mid m^s \in M^S\}$ as the set of all these elements. 
\begin{definition}
\label{def:strategy}
A \emph{strategy} is a function $f: O(M^S) \rightarrow 2^{T_u}$ such 
that for each $o(m^s)$, for each $t\in f(o(m^s))$, $^\bullet t \in o(m^s)$. In other words, each transition selected by a 
strategy in the observable and stable part of a marking must 
be enabled in this part of the marking. 
\end{definition}
\begin{example}
\label{ex:sature} 
Consider the net in Fig.~\ref{fig:satura}, without the grey places.  
Assume that the user can observe everything and his 
goal is to reach any marking among $\{p3, p5\}$, $\{p5, p7\}$ and $\{p4, p6\}$. 
A way for the user to reach this goal would be to wait for the environment to move, and to 
fire $t3$ if transition $t1$ occurred, and $t4$ otherwise. 
This would not be possible with a strategy defined on the stable parts of the markings of the 
initial net: the stable part of $\{p1, p2\}$, $\{p2, p3\}$ and $\{p2, p7\}$ is $\{p2\}$, therefore 
the strategy could not differentiate between the initial marking ($\{p1, p2\}$) and the markings reached after 
the occurrence of $t1$ ($\{p2, p3\}$ and $\{p2, p7\}$). 
This problem is overcome when we consider the extended net, with the place $p3 \lor p7$ (added 
in grey in the figure). 
Once that $t1$ occurs, this place will belong to the stable parts of all the markings subsequently 
reached, therefore the strategy can use the information that one of those two places must be 
marked to select transition $t3$. 
Note that the net in Fig.~\ref{fig:satura} is not the complete extended net, but the represented 
places are sufficient for the sake of presentation.
\begin{figure}
    \centering
    \includegraphics[width=0.7\textwidth]{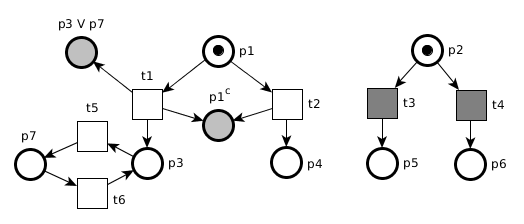}
    \caption{A net and two of its implicit places (in grey)}
    \label{fig:satura}
\end{figure}
\end{example}
The aim of the following example is to show the importance of 
allowing that some places are not observable, even when they 
belong to a stable marking.
\begin{example} 
\label{ex:rapina}
Consider the net in Fig.~\ref{fig:rapina}, where the user 
represents a robber and the environment a banker. 
For the sake of clarity the net is not extended; 
however this is irrelevant for the following discussion. 
The goal of the user is to reach place $w$, representing a 
successful robbery, without ending in place $p$, representing 
the prison. 
In the initial marking, the environment can set the code of 
the vault to one or to two (by executing $sc1$ or $sc2$ 
respectively). 
In order to successfully end the robbery, the user 
needs to choose between $i1$ and $i2$, agreeing with the 
decision of the environment. 
If the robber observed everything, he would not 
need to communicate with the environment in 
order to know the code; he could directly go 
to the vault (through transition $gv$) and observe whether 
$s1$ or $s2$ is marked. 
This happens because $s1$ (resp. $s2$) is in the stable part 
of the markings $\{s1, v\}$ (resp. $\{s2, v\}$), 
since there are no uncontrollable transitions enabled in 
those markings. 
However, this scenario is not realistic, 
and we should assume $s1$ and $s2$ are not visible 
by the robber.

\begin{figure}
    \centering
    \includegraphics[width = 0.8\textwidth]{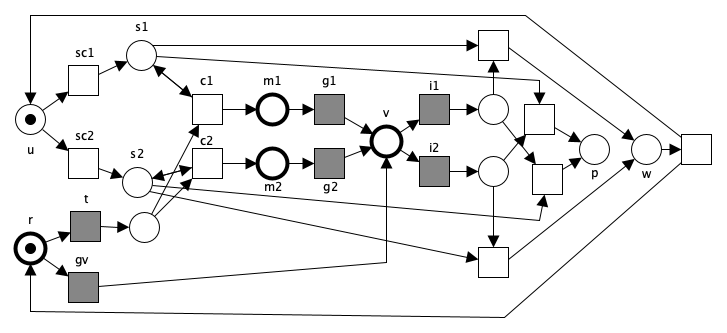}
    \caption{Net inspired by the example presented in \cite{Schobbens04}, where the user represents a robber, 
    and the environment a banker}
    \label{fig:rapina}
\end{figure}

\end{example}
So far we dealt with goals and plays informally, however in order to 
automatically check if the user has a winning strategy, we need 
to formally define them. 

Let $\Sigma$ be a net system and $\Sigma'$ its extension. 
Let $\unf(\Sigma') = (B, E, F, \mu)$ be the unfolding of $\Sigma'$, then the set $E$ can be partitioned into the set of uncontrollable events $E_{env} = \{e\in E \ |\  \mu(E) \in T_{env}\}$ and the set of controllable events $E_u = \{e\in E \ |\  \mu(E) \in T_u\}$.
We define a play as a run on $\unf(\Sigma')$, weakly fair with respect to the uncontrollable
events, together
with an increasing sequence of B-cuts, which can be seen as a potential
record of the play as observed by an external entity. Formally:
\begin{definition}
\label{def:play}
Let $\rho = (B_\rho, E_\rho, F_\rho, \mu_\rho)$ be a run on
$\unf(\Sigma')$ and $\delta = \gamma_0,\gamma_1,\cdots,$ $\gamma_i,\cdots$
an increasing sequence of B-cuts.
The pair $(\rho, \delta)$ is a \emph{play} iff:
\begin{enumerate}
  \item for each uncontrollable event $e$ in $\unf(\Sigma')$, the net obtained by adding
        $e$ and its postconditions to $\rho$ is not a run of $\unf(\Sigma')$;
  \item $\forall e\in E_\rho$ there is a B-cut  $\gamma_i\in\delta$
        such that $e<\gamma_i$. 
\end{enumerate}
\end{definition}
\begin{figure}
    \centering
    \includegraphics[width = 0.9\textwidth]{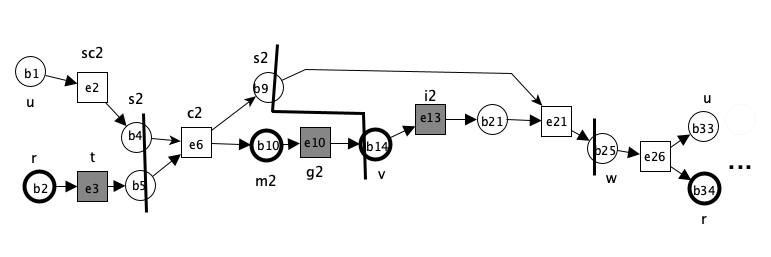}
    \caption{A play on the net in Fig.~\ref{fig:rapina}}
    \label{fig:play}
\end{figure}
Fig.~\ref{fig:play} represents a play on the net in Fig.~\ref{fig:rapina}. 

The \emph{winning condition} for a user is a set of plays. 
We are particularly interested in the case in which the winning 
plays satisfy a certain property, specified through a 
state based $\ltlx$ formula. 
$\ltlx$ is the fragment of LTL in which the X operator has 
been removed; other works, such as \cite{EH01,KPR07,SK04}, 
studied this logic on Petri nets. 
We identify the set of atomic propositions as the set of 
places of the net;
a proposition $p$ is true in a marking $m$ iff $p\in m$. 
\begin{example}
All the goals discussed until this point can be expressed through an $\ltlx$ 
formula. 
As an example, consider the net in Fig.~\ref{fig:play} which represents a play on the 
net discussed in Ex.~\ref{ex:rapina}. 
The goal of the user is represented by the formula $Fw \wedge  G \lnot p$
\end{example}
Given an $\ltlx$ formula $\phi$, and a play $(\rho, \delta)$, 
we can evaluate $\phi$ on a maximal refinement 
$\delta'$ of $\delta$ by considering the sequence of 
markings corresponding to the B-cuts in $\delta'$ and applying the usual semantics of LTL. A play $(\rho, \delta)$ is 
winning for the user when the goal is $\phi$ iff $\phi$ 
holds in all the maximal refinements of $\delta$. 

The following definition allows us to define a fairness constraint for the user. 
\begin{definition}
In a play $(\rho, \delta)$ the user is \emph{finally postponed} 
iff there is a controllable event $e\in E_u$ and a B-cut 
$\gamma_i\in \delta$ such that, for each $\gamma_j\in \delta$ such that $\gamma_i< \gamma_j$, it holds that $^\bullet e \subseteq \gamma_j$, and $\mu_\rho(e)$ is chosen by $f$ in the observable stable part 
of the marking associated with $\gamma_j$; in formula 
$\mu_\rho(e)\in f(o(\mu_\rho(\gamma_j)^s))$.  
\end{definition}
Since the game that we are considering is asynchronous, we 
focus our attention on the plays where the user is not finally 
postponed: if a transition is constantly enabled, the 
environment has no way to prevent at some point its occurrence.
\begin{definition}
A play $(\rho, \delta)$ is \emph{consistent with a strategy} $f$ iff: 
\begin{enumerate}
    \item For each controllable event $e\in E_u\cap E_\rho$, 
    there is a cut $\gamma_i\in\delta$ such that 
    $\mu_\rho(\gamma_i) = m_i$, $\mu_\rho(e)\in f(o(m_i^s))$, and $e$ is the only 
    event between $\gamma_i$ and $\gamma_{i+1}$. In other 
    words, each controllable event must be chosen by the 
    strategy in the observable stable part of a marking 
    that is compatible with a cut of the play. 
    \item The user is not finally postponed. 
\end{enumerate}
\label{def:cplay}
\end{definition}
A strategy $f$ is \emph{winning} for the user iff  there is at least
one play consistent with $f$, and
the user wins every play consistent with $f$, 
whatever the environment does. 
\begin{example}
Consider the net in Fig.~\ref{fig:satura} where the user has goal 
$F(p4 \wedge p5) \lor F((p3 \lor p7) \wedge p6)$. 
If the user observes everything, then there is a winning strategy $f$ 
defined in this way: $f(\{p2\}) = \emptyset, f(\{p2, p3 \lor p7\}) = \{t4\}, 
f(\{p2, p4\}) = \{t3\}$. 
If some of the places are unobservable there may not be a winning strategy. 
Consider $P_O = \{p1, p2, p5, p6\}$. In this case the user does not have a winning strategy: 
by observing the net he knows when $t_1$ or $t_2$ fired (because he will observe the place $p1^c$ of the saturated net), 
but cannot know which of them, and consequently he cannot decide which transition to fire in 
order to reach his goal.
\end{example}

\section{From Petri nets to concurrent game structures} 
\label{sec:pn2cgs}
In this section we show how our game relates to a \emph{turn-based asynchronous game} 
played on a concurrent structure. 
These structures are extensions of Kripke models, and were introduced 
in \cite{AHK02} for model-checking Alternating-time 
Temporal Logic (ATL). 
ATL is an extension of CTL, and expresses the fact that in a 
multi-agent system, a group of agents has a set of winning strategies 
to enforce a certain property on the system. 
The logic $\atls$ extends ATL in the same way that CTL* extends CTL. 
In particular, $\atls$ includes LTL, therefore the problem that we 
want to solve on the net, i.e. checking whether the user is able 
to enforce a behavior specified by an $\ltlx$ formula $\phi$, can be 
translated into the problem of verifying whether the system satisfies 
the corresponding $\atls$ formula $\langle\langle user\rangle\rangle\phi$. 
This fragment of $\atls$ was previously 
considered in \cite{JPSDM20} and denoted with $\atlx$. 
Originally, ATL and concurrent structures were defined in \cite{AHK02} by assuming 
full memory and full observability. 
Subsequent works, such as 
\cite{Schobbens04,BDL09}, adapted the logic and the 
structures to express also contexts with partial observability and 
no memory. 
We will refer to the definition of game structure provided in \cite{Schobbens04}. 

Our aim in discussing this relation is twofold: 
showing how our game can be interpreted in one of the most used 
frames to reason about multi-agent systems, 
and providing an intermediate step that we believe to be helpful for understanding  
the correctness of the method that we will propose in Sec.~\ref{sec:algo} to find a winning strategy. 
Let $\Sigma = (P, T, F, m_{in})$ be a net system, 
$T_u$ the set of controllable transitions, 
$P_O$ the set of observable places, and $\Sigma' = (P', T, F', m'_{in})$ the extension of $\Sigma$. 
We construct a turn-based asynchronous game  structure $G_\Sigma$ where three players interact: 
the user ($u$), the environment ($env$), and a scheduler ($s$). 
$G_\Sigma$ is a structure \emph{derived} from $\Sigma$ and obtained through the following steps.
\begin{enumerate}
    \item We construct the marking graph of $\Sigma$, $\MG(\Sigma)$, and 
    we interpret it as a Kripke model, where each local state is a proposition. 
    For each state $m\in [m_{in}\rangle$, a proposition $p$ is true
    in $m$ iff $p\in m$. 
    \item We compute all the regions of $\MG(\Sigma)$ associated with the unions and 
    the complements of observable places. 
    We update the set of propositions true in each state as 
    follows. 
    For each region $r$, for each state $m\in r$, we add the boolean 
    combination of places associated to $r$ as local proposition that is 
    true in $m$. 
    With this operation we obtain a Kripke model isomorphic to the marking graph 
    of the extended net $\MG(\Sigma')$. 
    In particular, the set of propositions that are true in some states of the model 
    coincides with $P'$.
    \item For each marking $m$, we compute its stable part $m^s$. 
    Let $P'_O$ be the set of observable places in $P'$.
    We define an equivalence relation $\sim_u$ between the markings in 
    this way: for each $m_1, m_2\in [m'_{in}\rangle$, $m_1 \sim_u m_2$ iff 
    $o(m_1^s)= m_1^s \cap P'_O = o(m_2^s) = m_2^s \cap P'_O$. Intuitively, $m_1\sim_u m_2$ if the 
    user cannot distinguish $m_1$ from $m_2$.
    \item From each state $m$, we compute which controllable transitions 
    are enabled in $o(m^s)$ and we remove from $\MG(\Sigma')$  
    the controllable transitions which are enabled in $m$ and not in $o(m^s)$.
    We denote as $M^r \subseteq [m'_{in}\rangle$ the subset of states that are 
    reachable from $m'_{in}$ once that this step has been executed. 
    Since some places may not belong anymore to any reachable marking, the associated 
    propositions are false in all the states of $G_\Sigma$. 
\end{enumerate}
An example of this construction is in Ex.~\ref{ex:pn2g}.
Then, $G_\Sigma = (\{u, env, s\}, M^r, m'_{in}, P', d, \tau, \sim_u)$, where 
$u$, $env$, and $s$ are the three players, i.e. the user, the environment 
and the scheduler respectively, and $d$ and $\tau$ are as follows.
\begin{itemize}
    \item For each $i\in \{u, env, s\}$, for each $m\in M^r$, 
    $d_i(m)$ is the set of \emph{moves} available in $m$ for 
    player $i$. Specifically, for each $m\in M^r$, $d_u(m)$ is the 
    set of controllable transitions outgoing from $m$, plus a 
    move remaining in the same state, from now on denoted with 
    $\epsilon$; $d_{env}(m)$ is the set of uncontrollable 
    transitions outgoing from $m$; if no such transition is enabled, 
    $d_{env}(m) = \{\epsilon\}$; $d_s(m) = \{u, env\}$. 
    For each state $m\in M^r$, a \emph{move vector} enabled in $m$ 
    is a tuple $\langle j_u, j_{env}, j_s \rangle$, 
    where $j_i\in d_i(m)$, $i\in \{u, env, s\}$.
    \item $\tau$ is the \emph{transition function}; for each 
    $m\in M^r$ and vector move $\langle j_u, j_{env}, j_s \rangle$ 
    enabled in $m$, $\tau(m, j_u, j_{env}, j_s)$ is the state  
    reached on $\Sigma'$ when transition $j_{j_s}$ occurs. 
\end{itemize}

\begin{example}
\label{ex:pn2g}
Let $\Sigma$ be the net in Fig.~\ref{fig:net&gs}, except for the grey places. 
Let $p1$ and $p2$ be the observable places. 
In the right part of the figure, we show $\MG(\Sigma)$ (in black and grey), 
and the structure of $G_\Sigma$ (in black). 
In this example, the extended net $\Sigma'$ has only two places more 
than $\Sigma$, represented in grey in the figure, associated with the 
complements of the two observable places.  
As observed in Sec.~\ref{sec:pn}, $\MG(\Sigma')$ is isomorphic to 
$\MG(\Sigma)$, therefore the right part of Fig.~\ref{fig:net&gs} represents 
it. In Table~\ref{tab:stablefig5}, each state of $\MG(\Sigma')$ is 
associated with its stable and observable part. 
As observed in Ex.~\ref{ex:stable}, there is no marking where $p2$ is 
stable. 
For this reason, a winning strategy for the user can never depend from the 
occurrence of $t2$, and when we construct $G_\Sigma$ we can remove all 
the occurrences of this transition as well as of the transitions causally depending from it (as explained in point 4 of the procedure 
above).
\begin{table}
    \centering    
    \caption{Table representing the stable and observable part of each state of the marking graph in Fig.~\ref{fig:net&gs}} 
    \label{tab:stablefig5}
    \begin{tabular}{c|c}
        States & Propositions\\
        \hline
        1, 2 & $p1$ \\
        
        5 & $p1$, $p2^c$\\
      
        3, 6, 8, 11 & $p1^c$\\
        
        4, 7, 9, 10, 12, 13 & $p1^c, p_2^c$\\
    \end{tabular}
\end{table}

\begin{figure}
    \begin{minipage}{0.4\textwidth}
    \centering
    \includegraphics[width = 0.9\textwidth]{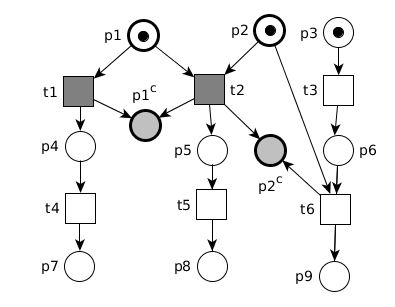}
    \end{minipage}
%    \hfill
    \begin{minipage}{0.5\textwidth}
    \centering
    \includegraphics[width = 0.8\textwidth]{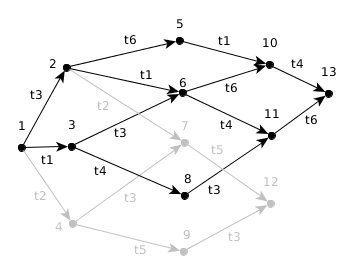}
    \end{minipage}
    \caption{A net system, its marking graph (black and grey) and its associated game structure (black)}
    \label{fig:net&gs}
\end{figure}
\end{example} 
An infinite computation $\lambda$ on $G_\Sigma$ 
is an infinite sequence of states starting from $m'_{in}$ such that the successor of each 
state is fully determined by the moves chosen in the previous one. 
In order to exclude the computations associated with unfair 
runs on the net,
we need to add to $G_\Sigma$ some \emph{weak fairness constraints}, as defined in \cite{AHK02}. 
A weak fairness constraint is a pair 
$\langle i, c\rangle$ where $i\in \{u, env, s\}$ and $c$ is a function 
that associates a set of moves for player $i$ to every state. 
Let $\lambda$ be an infinite computation and $m_k$ its $k-$th element; 
$\langle i, c\rangle$ is \emph{enabled} in $m_k$ if $c(m_k)\neq \emptyset$, 
and is \emph{taken} if there is a vector move $\langle j_u, j_{env}, j_s\rangle$ 
with $j_i\in c(m_k)$, $m_{k+1} = \tau(m_k, j_u, j_{env}, j_s)$, 
and if $i \neq s$, $j_s = i$. An infinite computation $\lambda$ is 
\emph{weakly fair} with respect to $\langle i, c\rangle$, if $\langle i, c\rangle$ 
is not enabled in infinitely many positions of $\lambda$ or it is taken infinitely 
many times. 

We require that there is no point in the computation from which the scheduler 
selects always the same player: this is expressed by the constraints 
$\langle s, c_u\rangle$ and $\langle s, c_{env} \rangle$, where 
$\forall m\in M^r$ $c_i(m) = \{i\}$, $i\in \{u, env\}$. In addition, we require 
that the computations are weakly fair with respect to uncontrollable transitions:
for each uncontrollable transition $t$, we define $\# t: M^r \rightarrow 2^{T_{env}}$ 
such that $\# t(m) = \emptyset$ if $t$ is not enabled in $m$, 
$\# t(m) = \{t\} \cup \{t_i : {\pre t_i}\subseteq m \, \wedge\, t_i \,\#\, t\}$ 
otherwise; for each $t$ uncontrollable, we add the constraint $\langle env, \#t \rangle$. 

For each state $m\in M^r$, let $[m]_u$ be the equivalence class of 
$m$ with respect to $\sim_u$, and $M_u$ be the set of all these classes. 
A strategy is a function $f: M_u \rightarrow T_u \cup \{\epsilon\}$. 
Given $\psi\in \ltlx$, a strategy is \emph{winning} with 
respect to the goal expressed by $\psi$ iff all the fair computations 
in which the user follows the strategy satisfy $\psi$. %$\phi$.

Given a computation $\lambda = m_{in}m_1...m_n...$ (finite or infinite) 
we can construct a new sequence $\pi(\lambda)$ in which all the 
consecutive states such that $m_i = m_{i+1}$ have been identified. 
Given an infinite sequence $\lambda$, $\pi(\lambda)$ can be also 
infinite, if there is no state $m_i\in \lambda$ such that for each 
$m_j : j > i$ $m_j = m_{j+1}$, or finite otherwise. 
\begin{lemma}
\label{prop: equiv}
Let $\psi$ be an $\ltlx$ formula, and $\lambda_1$ and 
$\lambda_2$ be two infinite computations such that 
$\pi(\lambda_1) = \pi(\lambda_2)$. 
$\lambda_1$ satisfies $\psi$ iff $\lambda_2$ does.
\end{lemma}
\begin{proof}
If the operator $X$ is not allowed, then the
validity of $\psi$ depends only on the sequence of distinct states
in computations; since, by hypothesis, $\pi(\lambda_1) = \pi(\lambda_2)$,
the thesis follows.
\end{proof}
\subsection{Relation between plays in the two models}
Let $\Sigma = (P, T, F, m_{in})$ be a net
system, $\Sigma'$ its extension, and $G_\Sigma$ the derived game structure. 

We now prove some propositions that show the relations between 
the plays on the unfolding and the infinite fair computations 
on the concurrent structure.

Let $\lambda=m_0m_1\cdots$ be a fair computation on $G_\Sigma$.
For every pair of consecutive states $m_i, m_{i+1}$, such that 
$m_i \neq m_{i+1}$, the move on the concurrent game structure is
associated, by construction, with a transition $t$ in $\Sigma'$, such that $m_i[t\rangle m_{i+1}$.

Given $\lambda$, we construct a run $\rho$ on the unfolding
of $\Sigma'$, starting 
from the initial cut and adding the events associated 
with the moves that bring from a state to the next one in the same 
order that the states have in $\lambda$.
At the same time we 
construct an increasing sequence $\delta$ of cuts: initially, $\delta = \gamma_0$; 
after an event $e$ is added to the run, we add to $\delta$ the cut 
$\gamma = (\gamma'\setminus {\pre{e}}) \cup \post{e}$, where 
$\gamma'$ is the last cut added in $\delta$ before $e$.
The pair derived in this way will be denoted
by $(\rho,\delta)[\lambda]$.
\begin{lemma}
\label{prop:pcgs2pn}
The pair $(\rho,\delta)[\lambda]$ is a play on $\unf(\Sigma')$.
\end{lemma}
\begin{proof}
The sequence of cuts $\delta$ satisfies condition 2 
of Def.~\ref{def:play} by construction. 

We need to prove that there is no uncontrollable event $e\not\in\rho$ 
such that $\rho \cup \{e\}$ is a run on the unfolding. 
By contradiction, assume that we can find such an event $e$. 
Then, there is $i\in \mathbb{N}$ such that, on $G_\Sigma$, the move $\mu(e) = t$ is 
enabled in every state $m_j\in \lambda : i<j$, but is never chosen. 
The reason cannot be that the environment player is never selected 
by the scheduler, otherwise, the constraint $\langle s, c_{env}\rangle$ 
would not be satisfied, and $\lambda$ would not be a fair computation. 
Hence the move $t$ is available in every state $m_j$ for player $env$, but he never chooses it. 
Then, the computation does not satisfy the fairness constraint 
$\langle env, \#t \rangle$. 

Hence, $(\rho, \delta)[\lambda]$ constructed in this way is a play
on the unfolding.
\end{proof}
With a similar idea, for each strategy $f$ as defined in Def.~\ref{def:strategy},  
we can associate every play on $\unf(\Sigma')$ consistent with 
$f$ with a set of infinite fair computations on $G_\Sigma$: 
let $(\rho, \delta)$ be a play on the unfolding; 
as first step, we consider all the possible sequentializations of events
included between the initial cut $\gamma_0$ and the next 
cut $\gamma_1\in \delta$. For each of these linearizations we can 
find a prefix of a computation on the concurrent game structure by executing 
on it, from the state corresponding to the initial marking, 
all the events in the order given by the linearization; 
then, we extend all these prefixes by considering the 
successive pairs of consecutive cuts, and the sequentializations 
of the events between them. 
If the play on the distributed net system reaches a deadlock, then the 
only possibility in the associated concurrent game structure is to execute the 
transition that remains in the same state infinitely often. 

In this way, we obtain a set of infinite computations 
associated with a play $(\rho, \delta)$, 
denoted by $\Lambda(\rho, \delta)$. 
\begin{lemma}
Let $f$ be a strategy on $\Sigma'$, and $(\rho, \delta)$ be a 
play consistent with $f$ on $\unf(\Sigma')$.
For every computation $\lambda\in \Lambda(\rho, \delta)$ on $G_\Sigma$, 
there is at least one fair computation $\lambda': \pi(\lambda') = \pi(\lambda)$.
\label{prop:ppn2cgs}
\end{lemma}
\begin{proof}
The set of constraints of the scheduler must guarantee that no 
player is neglected forever. 
If $\rho$ is finite then every $\lambda \in \Lambda(\rho, \delta)$ is fair: 
the only move available for the environment is the empty move, that keeps the 
system in the same state; 
when the scheduler selects the user, he can decide to 
execute the move that keep the structure in the same state.

In addition, when $\rho$ is finite, also the constraints on the  
environment are satisfied: after a finite number of states there is 
no uncontrollable enabled transition, therefore those 
constraints are infinitely not enabled.

We now consider the case in which $\rho$ is infinite. 
Since a play on the unfolding must be fair with respect to 
uncontrollable transitions, in the computation $\lambda$ 
the fairness constraint $\langle s, c_{env}\rangle$ must be satisfied. 
If $\rho$ has infinitely many controllable events, 
then the fairness constraint $\langle s, c_u\rangle$ 
is also satisfied by construction; 
otherwise, since  the user has always 
the possibility not to move in the system, we can 
construct a sequence  $\lambda' : \pi(\lambda) = \pi(\lambda')$ 
with some repeated states, that represent points of the sequence in which the 
user was selected by the scheduler in $G_\Sigma$, but the 
state of the system did not change. 

The set of fairness constraints for the environment must be 
satisfied: by contradiction, we assume that there is a 
function $\#t$, such that $\langle env, \#t\rangle$ 
is not satisfied in $\lambda'$; this means that $\langle env, \#t\rangle$ 
is not enabled in a finite number of states in $\lambda$, but 
it is taken only a finite number of times. 
This means that there is an uncontrollable event $e\not\in\rho$ in the 
net that is enabled in all the cuts compatible with $\delta$, except 
for a finite set; since $\rho$ is infinite, there must be a 
cut $\gamma\in \delta$ such that $e$ is enabled in $\gamma$ and 
in all the cuts $\gamma_j \in \rho : \gamma_j > \gamma$; hence 
$\rho \cup \{e\}$ is a run on $\unf(\Sigma')$.
This is in contradiction with the hypothesis that $\rho$ is a play. 
\end{proof}
\begin{lemma}
Let $\psi$ be an $\ltlx$ formula.  
A computation $\lambda$ satisfies $\psi$ iff $(\rho, \delta)[\lambda]$ does. 
Analogously, a play $(\rho, \delta)$ satisfies $\psi$ iff 
all the computations in $\Lambda(\rho, \delta)$ do.
\end{lemma}
\begin{proof}
By construction, $\delta$ is a maximal refinement, 
therefore $(\rho, \delta)[\lambda]$ allows a unique 
sequentialization equivalent to the computation $\lambda$. 

Each element $\lambda\in \Lambda(\rho, \delta)$  is associated with 
a maximal refinement of $\delta$ by construction. 
If all the elements in $\Lambda(\rho, \delta)$ satisfy $\psi$, then also all the refinements of 
$\delta$ do. Viceversa, if there is a $\lambda\in \Lambda(\rho, \delta)$ that does not satisfy $\psi$, 
also the $\delta'$ maximal refinement of $\delta$ 
associated with $\lambda$ does not.
\end{proof}
\subsection{Relations between winning strategies in the two models}%
\label{sec:relstrat} 
Let $f': O(M^S) \rightarrow 2^{T_u}$ be a strategy for the user on  
$\Sigma'$.
We call $f_{f'}$ a strategy on $G_\Sigma$ constructed 
from $f'$ as follows.
\begin{enumerate}
    \item If $f'(o(m^s)) \neq \emptyset$, and $t$ is 
    a transition in $f'(o(m^s))$ arbitrarily chosen, 
    then $f_{f'}([m]_u) = t$.
    \item $f_{f'}([m]_u) = \epsilon$ otherwise.
\end{enumerate} 
Analogously, let $f: M_u \rightarrow T_u \cup \{\epsilon\}$ be a  
strategy on $G_\Sigma$; 
we define the derived strategy $f'_f$ on $\Sigma'$ as follows:
\begin{enumerate}
    \item $f'_f(o(m^s)) = \{f([m]_u)\}$, if $f([m]_u) \neq \epsilon$.
    \item $f'_f(o(m^s)) = \emptyset$, otherwise.
\end{enumerate}
\begin{theorem}
Let $\psi$ be an $\ltlx$ formula, $\Sigma$ a net system, $\Sigma'$ its extension,
and $G_\Sigma$ the derived game structure. 
A winning strategy exists on $\Sigma'$
for $\psi$ 
iff a winning strategy exists on $G_\Sigma$.
\end{theorem}
\begin{proof}
As first step, we show that if $f'$ is a  
winning strategy on $\Sigma'$ % $\unf(\Sigma')$ 
for $\psi$, then $f_{f'}$ is a 
winning strategy for $\psi$ on $G_\Sigma$.

Let \emph{out}$(m'_{in}, f_{f'})$ be the set of fair computations starting 
from $m'_{in}$ that the user enforces by following the 
strategy $f_{f'}$.
For any $\lambda\in$\emph{out}$(m'_{in}, f_{f'})$, we show that 
$(\rho, \delta)[\lambda]$ satisfies the conditions of Def. 
\ref{def:cplay}, and therefore is consistent with
$f'$: (1) let $e\in E_u\cap \rho$ be an event controllable by the user. 
By construction, $\delta$ is a maximal refinement, 
hence there must be two cuts $\gamma_j$ and $\gamma_{j+1}$ such 
that $e$ is the only event between them. By construction, if 
$e$ was added to the run after $\gamma_j$, there must be a state 
$m\in \lambda : \mu(\gamma_j) = m$ such that $f_{f'}([m]_u)) = \mu(e)$. 
By construction of the strategy, $\mu(e)\in f'(o(m^s))$.
(2) By contradiction, assume that the user is finally 
postponed; then there is a cut $\gamma\in \delta$ such that for each cut 
$\gamma_j$ compatible with $\delta$ such that 
$\gamma_j > \gamma$, $f'(o(\mu(\gamma_j)^s))\neq \emptyset$. By construction, 
there must be a state $m_i\in\lambda : \mu(\gamma) = m_i$ such that 
$f([m_j]_u) \neq \emptyset$ for each $m_j, j > i$; then $\lambda$ 
cannot be fair with respect to $\langle s, c_{u}\rangle$. 
$(\rho, \delta)[\lambda]$ is a play, and it is consistent with the 
strategy by construction, hence it is a winning play for the 
user and it satisfies the property expressed by $\psi$. 
Hence also $\lambda$ satisfies it on the concurrent game structure.

As second step, let $f$ be a  
winning strategy on $G_\Sigma$. 
We show that $f'_f$ is a winning strategy on $\Sigma'$.

Let $(\rho, \delta)$ be a play consistent with $f'_f$.
By contradiction, we assume that the formula 
$\psi$ is not satisfied on $\unf(\Sigma')$, 
and that the play $(\rho, \delta)$ testifies it. 
Then there must be a maximal refinement $\delta'$ of $\delta$ 
that does not satisfy the formula. 
$\delta'$ is associated with an infinite sequence of states 
$\lambda$ on $G_\Sigma$. 
If $\lambda$ is a fair computation on $G_\Sigma$, 
then it is also a computation consistent with the strategy: 
let $t$ be a user's move in $m_i\in \lambda$. 
By construction, there must be $e: \mu(e) = t$ 
in $(\rho, \delta)$ such that $t\in f(o(\mu(\gamma_i)^s))$, where $\gamma_i\in \delta'$ and $\mu(\gamma_i) = m_i$. 

If $\lambda$ is not fair, we know by Lemma~\ref{prop:ppn2cgs} 
that we can construct a fair computation $\lambda'$ 
such that $\pi(\lambda) = \pi(\lambda')$. 
In addition, from the proof of Lemma~\ref{prop:ppn2cgs} 
we know that the only reason why $\lambda$ can be unfair 
is that the user is finally neglected 
during the play. This cannot happen if the user  has 
the strategy finally not-empty, otherwise the play on 
the net would not be consistent with this strategy 
(by definition the strategy selects only enabled events).
Then,  there is an infinite number of states in the sequence in which 
$f$ selects only the transition that keeps the system 
in the same state. We add a copy of these states in the 
position coming immediately after them. 
This computation $\lambda'$  is fair with respect to 
the user, therefore it is a play consistent with the 
strategy on the concurrent game structure and by hypothesis is winning. 
By construction $\pi(\lambda) = \pi(\lambda')$, hence 
Lemma~\ref{prop: equiv} guarantees that also 
$\lambda$ respect $\psi$.
\end{proof}
\section{Finding a winning strategy}
\label{sec:algo}
In this section we describe a method to find a winning strategy for the 
user on the net system, if one exists. 
In Sec.~\ref{sec:pn2cgs} we showed that this is equivalent to find a strategy 
on the derived game structure, and we mentioned that this problem can be 
translated into the problem of model-checking a $\atlx$ formula. 
However, to our knowledge, none of the model-checkers developed for ATL can be 
directly applied to solve the problem as we stated. 
This is because we want to model check a fragment of ATL*, whereas most of the 
tools only support ATL, and because even those tools 
supporting ATL* do not allow to specify the observability and the fairness 
constraints that we require. Among the most prominent tools, MOCHA \cite{AHMQ98} 
works only for synchronous systems and allows neither partial observability nor 
any kind of fairness constraint; MCMAS \cite{LQR17} does not support partial 
observability and consider only the so called \emph{unconditional fairness}, 
where each fair computation must cross infinitely often some subsets of states; 
finally the recent STV+ \cite{kurpiewski21} supports a special case of partial 
observability, in which the agents can observe only their local states, and does 
not support any fairness constraint. 

In this section we will focus on the research of a winning strategy on the derived 
game structure. 
The method that we propose can be summarized in the following steps, that will be explained in detail in the next subsection. 
\begin{enumerate}
    \item Among the set of potential winning strategies, 
    we select a strategy $f$ for $G_\Sigma$, and we prune $G_\Sigma$ accordingly. 
    \item From the pruned $G_\Sigma$ and the strategy 
    $f$, we construct a Kripke model $K(G_\Sigma, f)$ that allows us to verify 
    fairness constraints without having them externally specified.
    To the same aim, we construct a formula $\psi'$ from $\psi$ such that 
    the user has a winning strategy for $\psi$ on $G_\Sigma$ if, and only if, 
    $\psi'$ is verified in the initial state on $K(G_\Sigma, f)$. 
    \item We check if the strategy is winning with the help of an
    LTL model-checker; if it is, we stop and return that strategy,
    otherwise we mark the current strategy and all those equivalent
    to it as checked, and we repeat the steps.
\end{enumerate}
We say that two strategies $f$ and $f'$ are \emph{equivalent} if 
the set of states reachable from the initial state of 
$K(G_\Sigma, f)$ is equivalent to the set of states reachable from 
the initial state of $K(G_\Sigma, f')$. 
This may happen because some states become unreachable after the 
pruning process. In this case, $f$ is winning iff $f'$ is.
\subsection{Encoding fairness constraints into the Kripke model}
Let $G_\Sigma = (\{u, env, s\}, M^r, m'_{in}, P', d, \tau, \sim_u$) 
be a derived game structure with fairness constraints, and 
$f$ be a strategy on $G_\Sigma$. 
We modify $G_\Sigma$ and construct a Kripke model $K(G_\Sigma, f)$ as follows. 
Let $P^f$ be the set of atomic propositions of $K$. 
$P^f$ includes all the elements in $P'$. In addition, for each
transition $t\in T_{env}$, we add an element to $P^f$. Abusing the notation, these 
propositions will be denoted with the name of their associated 
transitions. Finally $P^f$ includes the two propositions $env$ and $u$.  

For each state $m \in M^r$, we delete all the outgoing 
controllable transitions that are not selected by $f$ in $[m]$. 
We split all the other arcs by adding for each of them an 
intermediate state.  
Let $t$ be a transition outgoing from $m$ in $G_\Sigma$, 
and $l$ be the intermediate state to add. 
All the propositions true in $m$ and all those associated with uncontrollable
transitions not enabled at $m$ are true in $l$. 
In addition, if $t$ is controllable or $f([m]) = \emptyset$, 
$u$ is true in $l$; if $m$ does not enable any transition in 
$T_{env}$, $env$ is true in $l$; if $t\in T_{env}$ all the 
propositions in $\#t$ are true in $l$. 

Finally, for each state $m$ in $M^r$ if $m$ does not enable any 
transition in $T_{env}$ and $f([m]) = \emptyset$, we add a state 
$l$ and two transitions so to create a cycle between $m$ and $l$. 
In $l$ the true propositions are those true in $m$ and all those 
in $P^f\setminus P'$. 
\begin{example}
Consider the net system  and its concurrent game structure represented 
in Fig.~\ref{fig:net&gs}. Assume that the goal of the user is $F(p7 \wedge p9)$, 
and that the user follows the strategy $f$ such that $f(\{p_1, p2^c\}) = \{t1\}$, 
$f(o(m^s)) = \emptyset$ otherwise. The Kripke model $K(G_\Sigma, f)$ encoding this 
strategy and the fairness constraints is represented in black in Fig.~\ref{fig:reduct}. 
The grey part represents the parts of the game structure $G_\Sigma$ that must be deleted. 
Each transition in $K(G_\Sigma, f)$ is doubled with respect to the same transition in 
$G_\Sigma$. The states dividing each transition ($1', 2', 5'$ etc.) are used to store 
the information about fairness constraints.  
The true propositions in each of these new states are in Table~\ref{tab:stabstr}. 
We can use this Kripke model to check if $f$ is winning, as formalized in Lemma~\ref{lem:fairness}. 
\begin{table}[]
    \centering
    \caption{True propositions in the states added during the construction of $K(G_\Sigma, f)$}
    \begin{tabular}{c|c}
    New states & True propositions\\
    \hline
        $1'$ & $p1, p2, p3, u, env, t3, t4, t5, t6$ \\
        $2'$ & $p1, p2, p6, u, env, t3, t4, t5, t6$\\
        $5'$ & $p1, p2^c, p9, u, t3, t4, t5, t6$\\
        $10'$ & $p1^c, p2^c, p4, p9, u, env, t3,t4,t5,t6$\\
        $13'$ & $p1^c, p2^c, p7, p9, u, env, t3,t4,t5,t6$
    \end{tabular}

    \label{tab:stabstr}
\end{table}

\begin{figure}
    \centering
    \includegraphics[width = 0.6\textwidth]{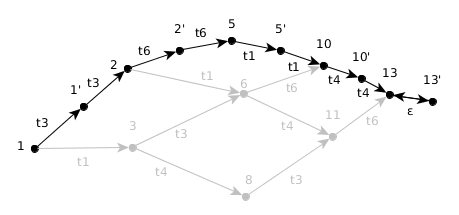}
    \caption{A Kripke model derived from the game structure in Fig.~\ref{fig:net&gs}}
    \label{fig:reduct}
\end{figure}
\end{example}
\begin{lemma}
\label{lem:fairness}
Let $\psi$ be an $\ltlx$ formula, $G_\Sigma$ a derived structure, 
$f$ a strategy and $K(G_\Sigma, f)$ the derived Kripke model. 
$f$ is winning on $G_\Sigma$ if, and only if, 
$\psi' = (\bigwedge_{x \in \{u, env \}}GF\,x \wedge \bigwedge_{t_j \in T_{env}}GF t_j) \rightarrow \psi$ holds in 
the initial state on $K(G_\Sigma, f)$.
\end{lemma}
\begin{proof}
We associate to every infinite computation $\lambda^f$ on $K(G_\Sigma, f)$ an infinite computation $\lambda$ on 
$G_\Sigma$ by removing from $\lambda^f$ all the states that are not in $M^r$.

We first show that for each infinite path $\lambda^f$ in $K(G_\Sigma, f)$,
$\lambda^f$ satisfies 
\[
    (\bigwedge_{x \in \{u, env \}}GF\,x \wedge \bigwedge_{t_j \in T_{env}}GF t_j)
\]
iff there is a computation 
$\lambda'$ in $G_\Sigma$ such that $\pi(\lambda) = \pi(\lambda')$, and $\lambda'$ 
satisfies all the fairness constraints in $G_\Sigma$. 
Assume such a $\lambda'$ exists, then both the user and the 
environment must have been selected infinitely often by the scheduler. 
If $\lambda'$ is obtained through the occurrence of  
infinitely many uncontrollable and controllable transitions, 
then by construction $\bigwedge_{i \in \{u, env \}}GF i$ holds in $\lambda^f$. 
If there is a finite number of uncontrollable 
transitions, there must be an index $i$ such that 
for each $\lambda'_j: j> i$, $\lambda'_j$ does not enable any 
uncontrollable transition. Since $\pi(\lambda) = \pi(\lambda')$, this property must hold also for $\lambda$. 
By construction, for each state in $\lambda$ where no uncontrollable transition 
is enabled, there is a state in $\lambda^f$ where $env$ is true. 
Analogously for $u$ if no controllable transition is 
enabled from a certain state on. 
In addition, there may be a finite number of controllable transitions also if 
the strategy for the user does not select any transition 
from a certain state on. By construction, for each 
of these states in $\lambda$ there is a state on $\lambda^f$ where $u$ is true. 
Vice versa, if there is no fair $\lambda'$, also the formula $\bigwedge_{i \in \{u, env \}}GF i$ does 
not hold on $\lambda^f$, since there are no 
other cases than those discussed above in which $u$ 
or $env$ are true in a state.
If $\lambda$ satisfies also the constraints for the environment, then the formula $\bigwedge_{t_j \in T_{env}}GF t_j$ must hold in $\lambda^f$, because 
by construction, for each $t_j\in T_{env}$, $t_j$ is true in a state 
$l$ of $\lambda^f$ iff there is a state $m$ in $\lambda$ and a transition $t_i$ bringing to the next state of $\lambda$ and 
such that $t_j$ is not enabled at $m$ (and therefore its associated fairness constraint is also not enabled) 
or $t_j \in \#t_i$ (and therefore the constraint is taken in $\lambda$). 

As second step, we observe that $\psi$ holds 
in $\lambda$ iff it holds in $\lambda^f$ (since all 
the propositions of $\psi$ are in $P$, this can be seen as a Corollary of Lemma~\ref{prop: equiv}). 

Then, if $\psi'$ holds in $K(G_\Sigma, f)$, in all the fair paths 
in $G_\Sigma$ that follow $f$, $\psi$ holds, and 
therefore $f$ is winning on $G_\Sigma$. 
Vice versa, if $\psi$ does not hold in $K(G_\Sigma, f)$, there is at least a 
fair path on $G_\Sigma$ that follows $f$ and such 
that $\psi$ is not verified, therefore $f$ is not winning on $G_\Sigma$
\end{proof}
Lemma~\ref{lem:fairness} allows us to prove the correctness of the algorithmic procedure proposed at 
the beginning of this section to find a winning strategy, if one exists. 
\subsection{Complexity}
\label{sec:compl}
\begin{theorem}
The complexity of finding a winning strategy is at least PSPACE and 
at most EXPSPACE in the dimension of the net. 
\end{theorem}
\begin{proof}
In order to prove that finding a winning strategy is at least PSPACE, we 
show that this problem is at least as hard as the reachability problem. 
This consists in deciding whether a given marking is reachable from the 
initial marking, and it is proved to be PSPACE-complete for 1-safe nets \cite{CEP95}. 
Let $\Sigma = (P, T, F, m_0)$ be a net system, and $m \subseteq P$ be a marking. 
We associate to the reachability problem for $m$ from $m_0$ a game on the net. 
On $\Sigma$, we assume that all the places are observable, and all the transitions 
are controllable by the user. Since every transition is controllable, for 
each marking, the stable part is the entire marking, therefore the 
strategy is defined on the reachable markings of the net. 
Let the goal of the user be $F(\bigwedge_{p_i\in m}p_i \wedge \bigwedge_{p_i\in P\setminus m} \lnot p_i)$. 
A winning strategy for the user exists iff $m$ is reachable: 
\begin{itemize}
    \item If $m$ is reachable, then there is a sequence of transitions that 
    brings from $m_0$ to $m$. Let $\sigma = t_0... t_n$ be this sequence, 
    and $m_0m_1...m_n$ be the corresponding sequence of markings, so that 
    $m_0[t_0\rangle m_1[t_1\rangle...m_n[t_n\rangle m$. 
    We can assume that for each $i\neq j$, $m_i\neq m_j$. If this is not the 
    case, we can construct another path that from $m_0$ arrives to $m$ 
    by observing that if $m_l = m_{l+k}$ for some $l, k$, then the sequence 
    of transitions 
    $t_0t_1...t_{l-1} t_{l+k}...t_n$
    also reaches $m$ and the 
    repetition between $m_l$ and $m_{l+k}$ has been deleted. 
    Then a winning strategy is a function $f$ such that 
    $f(m_i) = \{t_i\}$ for each $i\in \{0,...,n\}$.  
    \item If there is a winning strategy for the user, then by construction  
    the user reaches $m$ after a finite number of steps. 
\end{itemize}
We now prove the upper limit by showing that the algorithm that 
we proposed is at most EXPSPACE. 
In general, the marking graph $\MG(\Sigma)$ has exponential size 
with respect to the size of $\Sigma$; computing all the unions 
and complements of the observable places is also exponential with 
respect to $\Sigma$.
Computing the stable and observable parts of all the markings 
requires polynomial time with respect to $\MG(\Sigma)$. 
Let $G_\Sigma$ be the concurrent structure derived from $\Sigma$. 
By construction, $\MG(\Sigma)$ and $G_\Sigma$ have the same size. 
In our algorithm we guess a strategy $f$, we construct a Kripke 
model $K(G_\Sigma, f)$ and we verify an LTL formula.
Since constructing a strategy $f$ and $K(G_\Sigma, f)$ is 
polynomial with respect of $G_\Sigma$, and model-checking of LTL-X 
is PSPACE, the algorithm is PSPACE with respect to the size of 
$\MG(\Sigma)$. 
Since $\MG(\Sigma)$ is in general exponentially 
larger than $\Sigma$, the algorithm is at most EXPSPACE.
\end{proof}
\section{Related works}
\label{sec:art}
Among the works putting together temporal logic and asynchronous systems, a prominent line is 
represented by the works on ATL, already introduced 
in Sec.~\ref{sec:pn2cgs}. 
After its introduction in \cite{AHK02}, many authors  
proposed extensions for it, by providing 
interpretation of this logic with bounded or no memory (\emph{imperfect recall}), 
and partial observability (\emph{imperfect information}) \cite{Schobbens04,BDL09,LMS15}. 
Complexity and decidability results were proved for 
these semantics \cite{Schobbens04,LMO07}, and 
some heuristics where proposed to improve the performances of the algorithms in practice \cite{PBJ14}.
ATL is interpreted over \emph{concurrent game structures} (\cite{AHK02,Schobbens04}), or 
\emph{interleaved interpreted systems} (\cite{BPQR15,JPSDM20}). 
In Sec.~\ref{sec:algo} we shortly discussed  
some tools for its model checking. 

Other extension increasing the expressiveness of ALT 
includes the use of strategy contexts \cite{BDL09,LMS15} and probability \cite{KNPS20,CFKP13}.

The game on the unfolding defined in this paper 
develops some of the ideas presented in \cite{ABP21}. 
However, in \cite{ABP21}  we assumed 
the user to have full knowledge of the current global state of the system, and we defined the game 
only for a restricted class of 1-safe nets.  
In \cite{ABP21}, the goal of the user is to reach 
a target transition, and the algorithm to find a strategy 
is developed on a prefix of the unfolding. 

In \cite{BKP17} another game with a reachability goal on a Petri net is defined. 
In particular, the goal of the user consists in firing a 
target transition infinitely many times. 
Under the assumption of full observability, the authors reduce the game 
on the net to a Streett game on a graph. 

In \cite{FO17,GOW20} a different notion of game on the unfolding is presented. 
In this game, the players are represented as tokens on the net, and they are divided into two teams, the 
\emph{system} and the \emph{environment} players. 
The goal of the system players is to avoid a set 
of bad markings. Each token knows its own past, and is 
informed about the past of the others when it synchronizes with 
them through transitions in the net. 
Also in this case, in order to find a winning strategy, 
the authors translate their problem on a graph. 

\section{Conclusion}
\label{sec:concl}
In this paper we defined a game between two players, 
a user and an environment, on the unfolding of a 
Petri net. We proposed a notion of information 
available for the user that considers both that the 
value of some places may not be accessible at all for 
the user, and that the value of some places may change before 
the user can observe them. 

In the case where the user has no memory and must guarantee a property 
expressed with a $\ltlx$ formula, we provided a method 
to determine whether he has a winning strategy, 
i.e. he can enforce the system to satisfy the $\ltlx$ formula. 

Our definition of the strategy and of the user's observations 
keeps into account the presence of different components in the system, 
that may change their 
state, while the communication of a previous step reaches the user. 
This makes our model suitable to represent distributed system, 
where different components may be physically distant from each other. 
In future works we plan to exploit even more the concurrent 
structure of the net, by looking for strategies directly on the unfolding. 

In Sec.~\ref{sec:compl} we showed that the problem considered in 
this paper has a high complexity, that may create problems for 
practical uses. To address this problem, we will work in two directions: 
finding heuristics that allow for fast computations in most of the cases, 
and identifying subclasses of goals and structural properties of the net with lower complexity. 
We implemented a preliminary version of a model-checker tool and tested it on 
small examples. Not surprisingly, on bigger examples the tool spends most of the 
time testing different strategies. 
In order to improve the performance, we will focus on the research of more 
convenient orders to select the strategies to check. 

In addition, we intend to study the cases of bounded and full memory strategies. 
For these cases, we plan to define the memory of the user through the 
structure of the unfolding. 
We plan to find algorithms also for this case, and to compare the information of 
the user in this case with the information of a full memory strategy defined on 
a game structure. 

Finally, we are working at the synthesis of strategies 
by adding a controller component to the original system in order to 
restrict its behaviors, so that every 
execution satisfies the desired specification, in 
the same line of \cite{RW87}.
\acknowledgements
This work is partially supported by the Italian MUR.
\bibliographystyle{fundam}
\bibliography{abp22fi}

\end{document}